\documentclass[runningheads]{llncs}
\usepackage[T1]{fontenc}
\usepackage{graphicx}
\usepackage{float}
\usepackage{booktabs}
\usepackage{amsmath, amssymb} 
\usepackage[linesnumbered,ruled,vlined]{algorithm2e}
\usepackage{algpseudocode}

\usepackage{tikz}
\definecolor{color-a}{RGB}{64,143,191}
\definecolor{color-b}{RGB}{191,218,233}
\definecolor{color-c}{RGB}{195,179,214}
\definecolor{color-d}{RGB}{245,248,250}

\usepackage[misc]{ifsym}
\usepackage{tabularray}
\usepackage{rotating}
\usepackage{makecell}
\usepackage{xcolor}
\usepackage{adjustbox}

\usepackage{ntheorem}

\newtheorem{observation}{Observation}

\begin{document}
\title{The Fairness of Maximum Nash Social Welfare Under Matroid Constraints and Beyond
\thanks{This submission has been accepted by WINE 2024. \\ This work is supported in part by the National Natural Science Foundation of China
(Nos. 12171444, 12301418, 12471306) and Natural Science Foundation of
Shandong (No. ZR2022QA014). }
}
\titlerunning{The Fairness of Max-NSW Under Matroid Constraints and Beyond}
%
\author{Yuanyuan Wang \and Xin Chen $^{(\textrm{\Letter})}$ \and Qingqin Nong}

%
\authorrunning{Y. Wang et al.}
%
\institute{Ocean University of China, Qingdao 266100,  China
\\ \email{wyy8088@stu.ouc.edu.cn, \{chenxin1403,qqnong\}@ouc.edu.cn}
}

%
\maketitle              
\begin{abstract}

We study the problem of fair allocation of a set of indivisible items among agents with additive valuations, under matroid constraints and two generalizations: $p$-extendible system and independence system constraints. 
The objective is to find fair and efficient allocations in which the subset of items assigned to every agent satisfies the given constraint.
We focus on a common fairness notion of envy-freeness up to one item (EF1) and a well-known efficient (and fair) notion of the maximum Nash social welfare (Max-NSW).
By using properties of matroids, we demonstrate that the Max-NSW allocation, implying Pareto optimality (PO), achieves a tight $1/2$-EF1 under matroid constraints. This result resolves an open question proposed in prior literature \cite{suksompong2021}. 
In particular, if agents have 2-valued ($\{1, a\}$) valuations, we prove that the Max-NSW allocation admits $\max\{1/a^2, 1/2\}$-EF1 and PO. 
Under strongly $p$-extendible system constraints, 
we show that the Max-NSW allocation guarantees $\max\{1/p, 1/4\}$-EF1 and PO for identical binary valuations. Indeed, the approximation of $1/4$ is the ratio for independence system constraints and additive valuations. 
Additionally, for lexicographic preferences, we study possibly feasible allocations other than Max-NSW admitting exactly EF1 and PO under the above constraints.


\end{abstract}
%
%
%


\section{Introduction}

Fair allocation of indivisible items among agents is a highly active problem in computational economics and artificial intelligence, due to its growing applications, e.g., Public Housing~\cite{deng2013}, Courses Allocation~\cite{budish2017},
the website of Spliddit (spliddit.org)~\cite{goldman2015}, and the company Fair Outcomes (fairoutcomes.com).
Central to this problem falls into two categories: fairness and efficiency. 
To capture fairness, envy-freeness (EF,~\cite{foley1967}) is a compelling criterion: every agent prefers her own bundle to another one's allocated bundle.
Unfortunately, an EF allocation does not always exist even for allocating one indivisible item between two agents. This motivates a series of less stringent notions of envy-freeness, e.g., envy-freeness up to one item (EF1). 
For general monotone valuations, an EF1 allocation can be computed in polynomial time~\cite{lipton2004}. 
For finding efficient EF1 allocations, we are interested in a well-known criterion of Nash social welfare (NSW) that calculates the product of agents' utilities. 
An allocation maximizing the Nash social welfare (Max-NSW) guarantees EF1~\cite{caragiannis2019} and PO.

While allocating indivisible items, we are more interested in relevant constraints, such as, cardinality constraints~\cite{biswas2018}, budget constraints~\cite{wu2021}, scheduling constraints~\cite{li2021}) and conflicting constraints~\cite{hummel2022}. 
These  categories of constraints can be formulated as special cases of independence system constraints. 
An independence system is a system $(E, \mathcal{F})$, where $E$ is a finite item set and $\mathcal{F}$ is a collection of (independent) subsets of $E$ that has hereditary property:  
if $D\in \mathcal{F}$ and $C\subseteq D$ then $C\in \mathcal{F}$.
As a powerful abstraction of independence, matroid structures~\cite{mestre2006,gourves2014,gourves2019,dror2021} have played a prominent role in combinatorial optimization. An independence system $(M,\mathcal{F})$ is a matroid if $\mathcal{F}$ also has augmentation property: for  $C, D \in \mathcal{F}$ with $|C| < |D|$, there is an item $x\in D\setminus C$ such that $C + x \in \mathcal{F}$~\footnote{The notation $C + x $ means $C \cup \{x\}$, likewise $C - x$ means $C \setminus \{x\}$.}. 
In particular, cardinality constraints are equivalent to partition constraints~\cite{biswas2018}. 

Due to the flexible properties of matroid (or independence system) structures, it is challenging to explore fair and efficient allocations under matroid constraints (or more general independence system constraints). 
Indeed, Suksompong~\cite{suksompong2021} illustrated that the existence problem of EF1 and PO allocation remained open under matroids constraints (in Section 4). 
In this paper, we are interested in making efforts towards this direction by Max-NSW allocations, approximate EF1 allocations, properties of matroids and other effective techniques.


\subsection{Our Contribution}
We have strived to examine the existence of Max-NSW and approximate EF1 allocations under three prominent constraints: matroids, $p$-extendible systems, and independence systems (the relationships between the three can be seen in Figure~\ref{fig1}). Throughout, we select a specific Max-NSW allocation guaranteeing Pareto Optimality (PO) under all three classes of constraints. 
The main contributions in this paper are summarized in Table~\ref{table1}.

\subsubsection{Matroid Constraints. (Section 3)} 
A system $(M, \mathcal{F})$ is a matroid if $M$ is a finite item set, and $\mathcal{F}$ satisfies that i) hereditary property: if $D\in \mathcal{F}$ and $C\subseteq D$, then $C \in \mathcal{F}$; ii) augmentation property: if $C, D\in \mathcal{F}$ and $|D|>|C|$, then there is an item $ x \in D\setminus C$ such that $C + x \in \mathcal{F}$. 
For general additive valuations, we illustrate that every Max-NSW allocation achieves $1/2$-EF1 under matroid constraints. 
This approximation ratio cannot be improved since we exemplify an instance that for arbitrary $\varepsilon > 0$, there is no feasible Max-NSW and $(1/2+\varepsilon)$-EF1 allocation under partition matroids. 
In particular, for identical valuations, Biswas et al.~\cite{biswas2018} have shown that an exact EF1 and PO allocation was guaranteed to exist by the Max-NSW allocation under matroid constraints. In this settings, we find an exact EF1 and PO allocation by leximin ordering. 
Furthermore, if all agents have binary valuations, Dror et al.~\cite{dror2021} computed an exact EF1 and PO allocation under partition matroids. 
We explore the extended settings of 2-valued valuations that each item's value falls into $\{1,a\}$ with $a > 1$ for each agent, and strikingly  demonstrate that every Max-NSW allocation is $\max\{1/a^2, 1/2\}$-EF1 and PO under partition matroids.

\subsubsection{Beyond Matroid Constraints. (Section 4)} 
The results focus on two general classes $p$-extendible systems and independence systems. 
We restrict the $p$-extendible system \cite{mestre2006} to a
strongly $p$-extendible system $(M, \mathcal{F})$ (Definition 7): 
if $C\in \mathcal{F}, D \in \mathcal{F}$ with $C \subset D$ and if $H\cap C=\emptyset$ such that $C \cup H\in \mathcal{F}$, then there exists a subset $Y \subseteq D \setminus C$ with $|Y| \leq p|H|$ such that $D \setminus Y \cup H\in \mathcal{F}$. 
Furthermore, we focus on strongly $p$-extendible systems, which still include all matroids.
For identical binary valuations, we show that every Max-NSW allocation achieves $ \max\{1/p, 1/4\}$-EF1 and PO under strongly $p$-extendible systems. 
Independence systems are systems that satisfy only hereditary property. For additive valuations, every Max-NSW allocation is $1/4$-EF1 
and the approximation is already tight. Interestingly, we further consider lexicographic preferences~\cite{saban2014}, under independence system constraints, and compute an exact EF1 and PO allocation by a greedy-method algorithm. 
\begin{table}
\small

\begin{talltblr}[
    caption={Summary of our main results. Identical, bi. and identical-bi refer to identical additive, binary additive, and identical binary additive valuations respectively.
    	LB and UB refer to lower bounds and upper bounds respectively. ``$\Rightarrow$ PO'' refers to outcomes guarantee Pareto optimality.},
    label={table1}
]{
    width=\linewidth,
    row{1-2} = {gray!20, font=\bfseries},
    hline{1,Z} = 1pt, 
    hline{3,4,7,8} = {solid},
    hline{2} = {solid}
}
      
     \SetCell[r=2]{l,m} \makecell{Constraints} &\SetCell[r=2]{l,m} \makecell{Valuations} & \SetCell[c=2]{c,m} Max-NSWs ($\Rightarrow$ PO)  &      & \SetCell[r=2]{c,m}  Others ($\Rightarrow$ PO) \\
                                    &             & \SetCell{c,m} LB     & \SetCell{c,m} UB                 & \\
    \SetCell[r=1]{c,m} \makecell{Partition \\ Matroids}    & \SetCell{c,m} additive   & \SetCell{c,m} $\frac{1}{2}$-EF1    &    $(\frac{1}{2}+\varepsilon)$-EF1    & 1-EF1~(bi.)\cite{dror2021}             \\ 
    \SetCell[r=3]{c,m} Matroids     & \SetCell{c,m} additive    & \SetCell{c,m} $\frac{1}{2}$-EF1    & $(\frac{1}{2}+\varepsilon)$-EF1  &  \\
                                    & \SetCell{c,m} identical   & \SetCell{c,m} 1-EF1 \cite{biswas2018}       &                          & 1-EF1(by leximin)\\ 
                                    
                                    & \SetCell{c,m} $\{1,a\}$-valued   & $\max\{\frac{1}{a^2}, \frac{1}{2}\}$-EF1         &                          &  \\  
    \makecell{Strongly \\ $p$-Extendible\\Systems} & \SetCell{c,m} identical-bi  & $\max\{\frac{1}{p}, \frac{1}{4}\}$-EF1 &$(\frac{2}{3}+\varepsilon)$-EF1 \makecell{} & \\
    \SetCell[r=2]{c,m} \makecell{Independence \\ Systems} & \SetCell{c,m} additive & \SetCell{c,m} $\frac{1}{4}$-EF1  & $(\frac{1}{4}+\varepsilon)$-EF1 & \\
                                    & \SetCell{c,m}  lexicographic      &          &            &  1-EF1 (by RR alg.)\\
\end{talltblr}
\end{table}

\subsection{Related Work}
There is a vast literature on fairly and efficiently allocating indivisible items without constraints.  
In particular, Caragiannis et al.~\cite{caragiannis2016,caragiannis2019} originated that a Max-NSW allocation is both strikingly EF1 and approximately Maximin share for additive valuations. 
Barman et al.~\cite{barman2018} provided an efficient greedy algorithm to compute a Max-NSW allocation for binary valuations while finding Max-NSW allocations is APX-hard in general~\cite{lee2017apx}. 
Benabbou et al.~\cite{benabbou2021} showed that Max-NSW and leximin allocations both processed the EF1 property for matroid rank valuations. 
Amanatidis et al.~\cite{amanatidis2021} established that a Max-NSW allocation was always envy-freeness up to any item (EFX,~\cite{chaudhury2020}) for 2-valued valuations. 
Plaut and Roughgarden~\cite{plaut2020} EFX allocations and PO allocations were not compatible with general valuations. 
Recently, Feldman et al.~\cite{feldman2024} demonstrated an optimal tradeoff that for any $0 \leq \alpha \leq 1$, an $\alpha$-EFX allocation can guarantee $1/(\alpha + 1)$-Max-NSW for additive valuations. 

Our work is related to fair and efficient allocations in constrained settings. 
Wu et al.~\cite{wu2021} showed that under budget constraints, the Max-NSW allocation achieved $1/4$-EF1 with tight approximation ratio. 
We extend this positive result to more general constraints (independence systems). 
Biswas and Barman~\cite{biswas2018} firstly considered special matroid constraints.
They developed efficient algorithms to compute EF1 allocations under laminar matroids for identical additive valuations. 
If items are labeled by goods or chores, Shoshan et al.~\cite{shoshan2023efficient} provided a polynomial-time algorithm to compute feasible allocations under partition constraints guaranteed PO and EF up to one good and one chore. 
Dror et al.~\cite{dror2021} demonstrated the existence of EF1 allocations by devising algorithms on the settings of heterogeneous partition matroids, and $n$ agents with binary additive valuations (or two agents with general additive valuations). 
Besides, more literature studied other fairness allocations (e.g., Maximin share) under matroid constraints, which can be traced back to references~\cite{hummel2022maximin,gourves2019,hummel2022}.
Beyond the above special settings, we will focus on central constraints of general matroids, $p$-extendible systems and even independent systems in this paper, and we further explore the efficiency of feasible allocations.

\section{Preliminaries}

Given an instance $\mathcal{I} = (N, M, (v_i)_{i\in N})$, where $M=\{g_1,\ldots, g_m\}$ is a finite set of items (or goods), $N = \{1,\ldots, n\}$ is set of agents. 
Each agent $i$ has an individual valuation function $v_i(\cdot): 2^{M} \rightarrow \mathbb{R}_{\geq 0}$.  
Throughout, we assume that $v_i(\cdot)$ is normalized: $v_i(\emptyset)=0$, and monotone: $v_i(S)\le v_i(T)$ for all $S \subseteq T\subseteq M$.
Our goal is to find fair and efficient allocations under broad constraints: matroids, $p$-extendible systems, and independence systems. 
The relationship between these three classes of constraints is in Figure~\ref{fig1}.

\subsection{Independence Systems and Matroids}

Let $\mathcal{M} = (E, \mathcal{F})$ be a set system including a finite set $E$ and a collection $\mathcal{F} \subseteq 2^{E}$ over a set $E$. An independence system satisfies hereditary property. 
\begin{definition}[Independence System]\label{def1} 
	A set system  $\mathcal{M} = (E, \mathcal{F})$ is an independence system if	
it satisfies the following properties:
	\begin{itemize}
		\item[(1)] $\emptyset \in \mathcal{F}$;
		\item[(2)] if  $D\in \mathcal{F}$ and $C\subseteq D$, then $C \in \mathcal{F}$;
	\end{itemize}
The elements of $\mathcal{F}$ are called independent set.
\end{definition}

Based on Definition~\ref{def1} and augmentation property, we define the central structure in combinatorial optimization, called matroids. 
\begin{definition}[Matroid]\label{def2}
An independence system
 $\mathcal{M}=(E, \mathcal{F})$ is a matroid if
 it satisfies the following property:
\begin{itemize}
	\item[(3)] if $C, D\in \mathcal{F}$ and $|D|>|C|$, then there exists $ x \in D\setminus C$ such that $C + x \in \mathcal{F}$. 
\end{itemize}
In addition, an independent set $B\in \mathcal{F}$ is a base if for all $x \in E \setminus B$, $B + x \notin \mathcal{F}$.
\end{definition}

Indeed, a base of $\mathcal{M}$ is a maximal independent subset of $E$.
Matroid structures has been extensively studied in combinatorial optimization. The following lemmas illustrate  key conclusions on matroids.

\begin{lemma}[\cite{brualdi1969comments}]\label{lemma1}
Let $\mathcal{M}=(E, \mathcal{F})$ be a matroid. If $B \in \mathcal{F}$ and $B'\in \mathcal{F}$ are two bases,
then there exists a bijection $\sigma: B \rightarrow B'$ such that 
$B + \sigma(z) - z$ and $B' + z - \sigma(z)$ are also bases in $\mathcal{F}$ for all $z \in B$. 
\end{lemma}

\begin{corollary} \label{cor1}
	Let $\mathcal{M}=(E, \mathcal{F})$ be a matroid.
	If $I, J\in \mathcal{F}$ are two independent sets with same cardinality, i.e., $|I| = |J| = k$,  then there exists a bijection $\sigma: I \rightarrow J$ such that 
	$I + \sigma(z) - z $ and $J + z - \sigma(z)$ are also in $\mathcal{F}$ for all $z\in I$. 
\end{corollary}

In fact, Corollary~\ref{cor1} is a direct result of applying Lemma~\ref{lemma1} to a special case of matroid, called the {\textit{truncated matroid}. To be detailed, given a matroid $\mathcal{M} = (E, \mathcal{F})$, $\mathcal{M}_k=(E, \mathcal{F}_k)$ is a truncated matroid of $\mathcal{M}$ if $\mathcal{F}_k=\{S\in \mathcal{F}: |S|\le k\}$. 
Furthermore, combining capacities and categories, we can define another special case of matroid, a {\textit{partition matroid}} $\mathcal{M} = (E, \mathcal{F})$ if $E = \{E_1, \ldots, E_\ell\}$ can be partitioned into $\ell\leq |E|$ categories with corresponding capacities $k_1,\ldots, k_\ell$, such that $\mathcal{F} = \{S\subseteq E: |S\cap E_i| \leq k_i, \forall i = 1,\ldots, \ell\}$. Specially, if there is a single category, a partition matroid is referred as a {\textit{uniform matorid}}. 
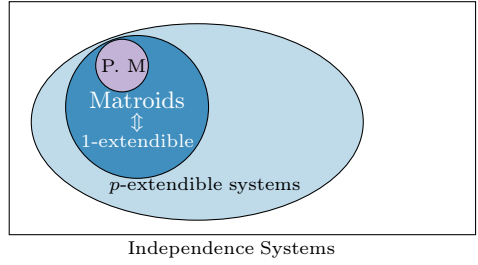
\begin{figure}[htp]
\centering
\def\r{1cm}
\begin{tikzpicture}
\draw[thick, black] (3.5*\r,0*\r) arc (0:360: 2.2cm and 1.3cm);
  \begin{scope}
    \clip (3.5*\r,0*\r) arc (0:360: 2.2cm and 1.3cm);
    \fill[color-b] (-1.2*\r,-1.5*\r) rectangle (5*\r,1.6*\r);
  \end{scope}
   \begin{scope}
    \clip (0.5*\r,0.2*\r) circle[radius=0.95*\r];
    \fill[color-a] (3.5*\r,0*\r) arc (0:360: 2.2cm and 1.3cm);
  \end{scope}
  \begin{scope}
    \clip (0.3*\r,0.75*\r) circle[radius=0.35*\r];
    \fill[color-c] (0.5*\r,0.2*\r) circle[radius=0.95*\r];
  \end{scope}
  \draw[black] (0.5*\r,0.2*\r) circle[radius=0.95*\r];
  \draw[black] (0.3*\r,0.75*\r) circle[radius=0.35*\r];
  \draw (-1.2*\r,-1.5*\r) rectangle (5*\r,1.6*\r);
  
  \node at (0.32*\r,0.75*\r) {\color{black}{\scriptsize P. M}};
  \node at (0.5*\r,0.3*\r) {\color{white}{\small Matroids}};
  \node at (0.5*\r,-0*\r) {\color{white}{\scriptsize $\Updownarrow$}};
  \node at (0.5*\r,-0.25*\r) {\color{white}{\scriptsize $1$-extendible}};
  \node at (1.4*\r,-0.85*\r) {{\scriptsize $p$-extendible systems}};
  \node at (1.4*\r,-1.1*\r) {{\scriptsize}};
  \node at (1.75*\r,-1.7*\r) {{\scriptsize Independence Systems}};
\end{tikzpicture}
\caption{Relations between several classes of constraints on independence systems. P. M refers to partition matroids.}\label{fig1}
\end{figure}

If we relax the augmentation property on matroids, a generalized structure can be derived, called $p$-extendible system.

\begin{definition}[$p$-extendible System~\cite{mestre2006}]\label{def3}
An independence system
$\mathcal{M}=(E, \mathcal{F})$ is a $p$-extendible system if it satisfies the following property:
\begin{itemize}
\item[(4)] if $C\in \mathcal{F}, D \in \mathcal{F}$ with $C \subset D$ and if $x\notin C$ such that $C + x \in \mathcal{F}$, then there exists a subset $Y \subseteq D \setminus C$ with $|Y| \leq p$ such that $D \setminus Y + x \in \mathcal{F}$. 
\end{itemize}
\end{definition}

\begin{lemma}[\cite{mestre2006}]\label{lemma2}
The system $(E, \mathcal{F})$ is a matroid if and only if is 1-extendible.
\end{lemma}

Based on the above lemma,  every matroid is also a $p$-extendible system for any $p\in N^+$. 
We remark that many natural problems fall in $p$-extendible framework, such as maximum weight $b$-matching, maximum profit scheduling, and matroid intersection.

\subsection{Fairness and Efficiency Criteria}


Given an independence system  $\mathcal{M} = (M, \mathcal{F})$ and an instance $(M, N, (v_i)_{i\in N})$. 
A {\textit{feasible allocation}} $X$ refers to a collection of $n$ disjoint bundles $X_1, \ldots, X_n$
such that $\cup_{i\in N} X_i \subseteq M$ and every bundle $X_i \subseteq M$ is an independent set in $\mathcal{F}$. Notice that all agents obey one independence system $(M, \mathcal{F})$. 
Due to independence constraints, a {\textit{complete allocation} $X$ (s.t. $\cup_{i\in N} X_i = M$) may not exist, and thus we admit incomplete allocations with unallocated items, i.e., $M \setminus (\cup_{i\in N} X_i) \neq \emptyset$. 
Unless otherwise specified, we will write ``an allocation'' instead of ``a feasible allocation'' in the remainder. 

In the independence settings (matroids, $p$-extendible or independence systems), our work is to explore ``fair'' and ``efficient'' allocations. 
The fairness notions we are interested in are based on envy-freeness. An allocation $X$ is {\textit{envy-free} (EF) if $X_i \in \mathcal{F}$ for every agent $i\in N$  and $v_i(X_i) \geq v_i(X_j)$ for every pair of agents $i$ and $j$. 
Due to independence constraints, a trivial but inefficient $A = (\emptyset, \ldots, \emptyset)$ is EF under $(M, \mathcal{F})$. We are also interested in a relaxed notion of classical EF property. 
\begin{definition}[EF1]
 Let $\alpha\in[0,1]$ be any constant.  An allocation $X$ is $\alpha$-approximate envy-free up to one item ($\alpha$-EF1) if  every bundle $X_i$ is in $\mathcal{F}$ and 
 \begin{equation*}
 \forall i, j \in N, \exists Z \subseteq X_j, |Z| \leq 1, v_i (X_i) \geq \alpha \cdot v_i (X_j \setminus Z).
 \end{equation*} 
In particular,  a 1-EF1 allocation is also called an exact EF1 allocation. 
\end{definition}

Besides the notion of EF1, we are also interested in the maximum Nash social welfare, which is a prominent notion of both fairness and efficiency.

\begin{algorithm}[h]
    \caption{The Max-NSW allocation}\label{alg1}
    
    {\bf Input:} an independence system $(M, \mathcal{F})$,  an instance $(N, M, (v_i)_{i \in N})$ \\
    
    {\bf Output:} A Max-NSW allocation $A^*$ \\
    
    Denote $N_X = \{i\in N: v_i (X_i) > 0\}$ for an allocation $X$; \\
    
  Select a largest agent set $N^* = \arg\max\{|N_{X}|: X \mbox{~with~} X_i \in \mathcal{F}, ~\forall i \}$; \\

  Denote by $\mathcal{T}_{N^*}$ collecting all Max-NSW allocations over $N^*$; 

  \If{$|\mathcal{T}_{N^*}| > 1$}
  {
  Select $T^*\in \mathcal{T}_{N^*}$ such that $|T^*| =\max\{|T|: T\in \mathcal{T}_{N^*}\}$;
  }
  \Else
  {
  $\mathcal{T}_{N^*} = \{T^*\}$;
  }

  $A^*_i \leftarrow T_i^*$ for all $i\in N^*$ and $A_i^* \leftarrow \emptyset$ otherwise;

    {\bf Return} $A^* = (A^*_1,\ldots, A^*_n)$  
\end{algorithm}

\begin{definition}[Max-NSW]
An allocation $A^*$ maximizes the Nash social welfare (NSW) if  
\begin{equation*}
NSW(A^*) = \max\{NSW(X): X_i\in \mathcal{F}, ~\forall i\},
\end{equation*}
where $NSW(X) = \prod_{i\in N} v_i (X)$ is the Nash social welfare of $X$.

In particular, we select the Max-NSW allocation returned by Algorithm~\ref{alg1}. 
\end{definition}
Specifically, Algorithm~\ref{alg1} firstly select a largest agent set $N^*$ that simultaneously received positive utilities. 
Then select a Max-NSW allocation $T^*$ that maximizes the number of agents in $N^*$. 
Finally, construct the output allocation by giving bundle $T_i^*$ to agents in $N^*$ and giving nothing to the remaining agents.

Based on Algorithm~\ref{alg1}, we show a crucial observation to capture that every Max-NSW allocation is PO. 

\begin{definition}[PO]
An allocation $A^*$ is Pareto optimal (PO) if no alternative allocation Pareto dominates $A^*$. Specifically, for every alternative $X$,
\begin{equation*}
\exists i\in N, v_i (X) > v_i (A^*_i) \Rightarrow \exists j\in N, v_j (X_i) < v_j (A_j^*).
\end{equation*}
\end{definition}
\begin{observation}\label{lemma3}
For general additive valuations, every Max-NSW allocation is also a PO allocation under independence system  (including matroid or $p$-extendible system) constraints. 
\end{observation}





\section{Matroid Constraints}

In this section, we explore the existence of EF1 and PO allocations under matroid constraints, which capture important properties, such as heredity and augmentation.
Recall that for binary valuations, Dror et al.~\cite{dror2021} computed an exact EF1 and PO allocation under partition matroid constraints. We will explore the results for general additive, 2-valued additive and identical additive valuations.

\subsection{General Additive Valuations}

For general additive valuations, we firstly state an negative result by constructing a worst instance under partition matroids. 
\begin{lemma}\label{lemma4}
    For any $\varepsilon > 0$, there exists an instance that no Max-NSW allocation can guarantee  $(\frac{1}{2} + \varepsilon)$-EF1 under partition matroid constraints. 

\end{lemma}

\begin{proof}
	Given an instance $\mathcal{I} = (N, M, (v_i)_{i\in N})$ with a matroid $\mathcal{M} = (M, \mathcal{F})$. 
	$M=\{g_1, \ldots, g_{2k+2}\}$ is an item set, and $\mathcal{F} = \{ S \mid |S| \leq k+1, S\subset M\} $ is a collection of independent sets. There are two agents $N = \{1,2\}$ with valuations $v_i(\cdot)$ as Table~\ref{table2}, where $k > 0$ is sufficiently large, and $\delta >0$ is sufficiently small.  
	\begin{table}
		\centering
		\caption{A worst instance  under matroid constraints.} \label{table2}
		\begin{tabular}{l  l l l l l l l}
			\toprule
			$g_j$     &$g_1$~~~~~   &$\cdots$~~~~~       &$g_{k+1}$ ~~~~~~~~~         & $g_{k+2}$     ~~~~~     &$\cdots$ ~~~~~ &$g_{2k+2}$~~~~~\\
			\midrule
			$v_1(g_j)$~~~~~~~ &$2k+1 - \delta$~~ &$\cdots$~~   &$2k+1 - \delta$~~      &$k$   &$\cdots$ & $k$ \\
			$v_2(g_j)$~~~~~~~ &$k$  &$\cdots$ &$k$      &$0$   &$\cdots$ & $0$ \\
			\bottomrule
		\end{tabular}
	\end{table} 
	Next, we will calculate Max-NSW allocations in this instance.
	Let $X $ be any feasible allocation under matroid $\mathcal{M}$, where $X = (X_1, X_2)$ and $|X_i| \leq k + 1$ for any agent $i = 1,2$.  Assume w.l.o.g that $X_1$ contain $x\in [0, k+1]$ items with the value of $2k + 1 - \delta$. 
	Due to $|X_i| \leq k +1$ for $i = 1,2$, we can get that $v_1(X_1) \leq (k + 1 - \delta) x + k(k + 1)$, 
	$v_2(X_2) \leq - kx + k (k+1)$, and $NSW(X) \leq \varphi(x) $, where $ \varphi(x) = - (k^2 + k - k\delta) x^2 + (k^2 + k - (k^2 + k)\delta) x + k^2 (k+1)^2$. Since that $\varphi(x)$ is decreasing over $[1, k+1]$ and $X\neq A^*$, we have that $NSW(X) < \max\{ \varphi(0), \varphi(1) \} = k^2(k+1)^2$.
	Thus, the allocation, that agent 1 gets $\{g_{k+2}, \ldots, g_{2k+2}\}$ and agent 2 gets $ \{g_1, \ldots, g_{k+1}\}$, is the unique Max-NSW allocation under matroid $\mathcal{M}$. 
	Let $A^* = (A^*_1, A^*_2)$ be an allocation with $A^*_1= \{g_{k+2}, \ldots, g_{2k+2}\}$ and $A^*_2 = \{g_1, \ldots, g_{k+1}\}$. 
	It holds that
	$$v_1(A^*_1) = k(k+1) < \left(\frac{1}{2} + \frac{1}{2k} \right) \cdot v_1(A_2 - g).$$ Define $\varepsilon = \frac{1}{2k}$, we can obtain the conclusion that no Max-NSW allocation can guarantee  $(\frac{1}{2} + \varepsilon)$-EF1 under partition matroid constraints. 
	\qed
\end{proof}

Now, we prove an important theorem of the section, which establishes an $\frac{1}{2}$-EF1 and PO allocation.

\begin{theorem}\label{th1}
For general additive valuations, every Max-NSW allocation guarantees $\frac{1}{2}$-EF1 and PO under matroid constraints.
\end{theorem}
\begin{proof}
Given an instance $\mathcal{I} = (N, M, (v_i)_{i\in N})$ with a matroid $\mathcal{M} = (M, \mathcal{F})$. 
$M=\{g_1, \ldots, g_m\}$ is item set and 
$N=\{1, \ldots, n\}$ is agent set.
Let $A=(A_1, \ldots, A_n)$ be a Max-NSW allocation.
  By Corollary 2, $A$ is a PO allocation. The following will show that $A$ is $1/2$-EF1. Suppose that there exist $i, j \in N$ such that 
\begin{equation}\label{eq1}
v_i (A_i) <\frac{1}{2} v_i(A_j - g), \forall g\in A_j. 
\end{equation}
We will discuss from two cases depending on whether $NSW(A)>0$ or not.

\noindent\textit{Case 1.}  $NSW(A)>0$. Due to $A$ is feasible under matroid $\mathcal{M}$, we have that $A_i \in \mathcal{F}$, $A_j \in \mathcal{F}$. 
If $|A_i| = |A_j|$, based on Corollary~\ref{cor1}, there exists a bijection
$\phi: A_i \to A_j$ such that $A_i - g + \phi(g) \in \mathcal{F}$ and $A_j - \phi(g)+g \in \mathcal{F}$.
 Otherwise if $|A_i| \neq |A_j|$, we can assume w.l.o.g. that $|A_i| < |A_j|$. Based on the definition of matroid, there must exist an item set $C=\{g_j^1, \ldots, g_j^l \} \subseteq A_j$ which suffices that $A_i \cup \{g_j^1, \ldots, g_j^l\} \in \mathcal{F}$ and $|A_i\cup \{g_j^1, \ldots, g_j^l\}| = |A_j|$. 
 We can also build a bijection from $A'_i = A_i \cup C$ to $A_j$, and every $g_0^k \in C$ is a copy item of $g_j^k$ with $v_i(g_0^k) = v_j (g_0^k) = 0$, $v_i(A_i)=v_i(A_i')$. Hence, we define $\phi:  A'_i \to A_j$ as a bijection from $A'_i$ to $A_j$. 

For any item pair $(g, \phi(g))\in \phi$, we define the value $\rho(g, \phi(g))$ as following: 
    $$
    \rho(g, \phi(g))=\frac{v_i(\phi(g))-v_i(g)}{v_j(\phi(g))-v_j(g)}. 
    $$
Observe that $\rho(g,\phi(g))=0$ if $v_i(\phi(g))-v_i(g) = 0$ and $v_j(\phi(g)) - v_j(g) \neq 0$, $\rho(g,\phi(g))=\infty$ if $v_j(\phi(g))-v_j(g)=0$ and $v_i(\phi(g)) - v_i(g) \neq 0$. We denote by $P=\{(g,\phi(g)), v_i(\phi(g))-v_i(g)>0\}$ and $Q=\{(g,\phi(g)), v_i(\phi(g))-v_i(g)\le 0\}$ the two types of sets of item pairs. 
For agent $i$, denote by $A_{iP}'=\{g\in A_i', \exists \phi(g)\in A_j, s.t. (g,\phi(g))\in P\}$ and $A_{iQ}'=\{g\in A_i', \exists \phi(g)\in A_j, s.t. (g,\phi(g))\in Q\}$ the item sets over $A'_i$.  
Similarly, for agent $j$, we denote $A_{jP}=\{\phi(g)\in A_j, \exists g\in A_i, s.t. (g,\phi(g))\in P\}$, $A_{jQ}=\{\phi(g)\in A_j, \exists g\in A_i', s.t. (g,\phi(g))\in Q\}$.
Let $(g^*,\phi(g^*))=\arg\max_{(g,\phi(g))\in P}{\rho(g,\phi(g))}$, we construct another allocation $B=(B_1, \ldots, B_n)$ where $B_i=A_i' + \phi(g^*) - g^*, B_j = A_j + g^* - \phi(g^*)$, and $B_k=A_k$ for any $k\ne i,j$. We will prove that 
\begin{align*}
v_i(B_i)v_j(B_j) - v_i(A_i)v_j(A_j)> 0.\tag{2}
\end{align*}
If the above Inequality~(2) holds, then $NSW(B) > NSW(A)$, which is contradicted with $NSW(A) = \max \{NSW(X) \mid X\mbox{~is~feasible}\}$. Thus, the following will focus on the proof of Inequality~(2).

Based on the definition of $A'_{iQ}$ and $A_{jQ}$, $v_i(A_{jQ})-v_i(A_{iQ}')\leq 0$, it holds that  
\begin{align*}\label{eq3}
      v_i(A_j)-v_i(A_i') & = [v_i(A_{jP})-v_i(A_{iP}')]+[v_i(A_{jQ})-v_i(A_{iQ}')]\\
        & \leq [v_i(A_{jP})-v_i(A'_{iP})].  \tag{3}
\end{align*}
Due to Inequality~(1), there must exist $(g,\phi(g))\in P$ such that $v_i(\phi(g))>v_i(g)$. It means that $(g,\phi(g))\in P$ and $\phi(g)\in A_{jP}$. Therefore, $(g^*,\phi(g^*))\in P$.

 If $v_j(\phi(g^*))-v_j(g^*)\le 0$,
  it is easy to verify that $NSW(B) > NSW(A)$, which is contradicted with that $A$ is a Max-NSW allocation. 
 
  If $v_j(\phi(g^*))-v_j(g^*)> 0$,
we can obtain that 
     \begin{align*}
        \frac{v_i(A_{jP})-v_i(A_{iP}')}{v_j(A_{jP})}&\ge\frac{v_i(A_{jP})-v_i(A_{iP}')}{v_j(A_j)}\ge \frac{v_i(A_{j})-v_i(A_i')}{v_j(A_j)}>\frac{v_i(A'_{i})+v_i(\phi(g^*))}{v_j(A_j)}.
     \end{align*}
The first inequality holds because $v_j(A_j)\ge v_j(A_{jP})$, the second inequality holds because of Inequality~(3), and the third inequality holds because of Inequality~(1).
Furthermore, by the definition of $(g^*, \phi(g^*))$, it holds that
      \begin{align*}
        \frac{v_i(\phi(g^*))-v_i(g^*)}{v_j(\phi(g^*))-v_j(g^*)}
        \ge \frac{v_i(A_{jP})-v_i(A_{iP}')}{v_j(A_{jP})-v_j(A_{iP}')}
        \ge\frac{v_i(A_{jP})-v_i(A_{iP}')}{v_j(A_{jP})}
        >\frac{v_i(A'_{i})+v_i(\phi(g^*))}{v_j(A_j)}.       
     \end{align*}
Notice that $v_i(\phi(g^*))-v_i(g^*)>0$ and $v_i(A_i')+v_i(\phi(g^*))>0$, we have 
    $$
    [v_i(\phi(g^*))-v_i(g^*)] \cdot v_j(A_j)>[v_i(A_i')+v_i(\phi(g^*))]\cdot[v_j(\phi(g^*))-v_j(g^*)].
    $$
Thus, 
     \begin{align*}
     v_i(B_i) & v_j(B_j)  - v_i(A_i)v_j(A_j)\\
     & =[v_i(\phi(g^*))-v_i(g^*)] \cdot v_j(A_j)-[v_i(A_i')+v_i(\phi(g^*))]\cdot[v_j(\phi(g^*))-v_j(g^*)]\\
     & ~~~~- v_i(g^*)\cdot [v_j(g^*)-v_j(\phi(g^*))] > 0.
     \end{align*}
\noindent\textit{Case 2.} $NSW(A)=0$.
If $v_i(A_i)\ge 0$ and $v_j(A_j)>0$, then the proof process is similar as Case 1.
Otherwise if $v_i(A_i)\ge 0$ and $v_j(A_j)=0$, we can construct another allocation $B=(B_1,\ldots, B_n)$ where 
$B_i=A_j$, $B_j=A_i$, $B_k=A_k$ for all $k\ne i, j$. It can be verified that $|\{i\in N: v_i (B_i) > 0\}|> |\{i\in N: v_i (A_i) > 0\}|$, which is contradicted with the definition of $A$. 

To be concluded, $A = (A_1, \ldots, A_n)$ is a $1/2$-EF1 and PO allocation. 
\qed
\end{proof}

\begin{corollary}\label{cor3}
For general additive valuations, every Max-NSW allocation guarantees $\frac{1}{2}$-EF1 and PO under partition matroid constraints.
\end{corollary}
Combining with Lemma~\ref{lemma4}, we remark that the approximation of $1/2$ is a tight ratio under partition matroid constraints.

\subsection{Special Cases of Additive Valuations}

This section considers two special cases of additive valuations: identical additive  and 2-valued additive  valuations. 
For identical additive valuations, Biswas et al.~\cite{biswas2018} showed that the Max-NSW allocation admitted 
exact EF1 and PO under matroid constraints.
We provide an exact EF1 and PO allocation in this settings by leximin ordering. 
\begin{theorem}\label{th2}
    For identical valuations, every leximin allocation guarantees exactly EF1 and PO under matroid constraints.
\end{theorem}
\begin{proof}
	
	Given an instance $\mathcal{I} = (N, M, v)$ with a matroid $\mathcal{M} = (M, \mathcal{F})$. 
	$M=\{g_1, \ldots, g_m\}$ is item set and 
	$N=\{1, \ldots, n\}$ is agent set. Notice that all agents share an identical valuation function $v(\cdot)$.
	Let $A = (A_1,\ldots, A_n)$ be a  leximin allocation such that $A_i\in \mathcal{F}$ for all $i$. By the definition of leximin,  $A \succeq X$ for any allocation $X$ and ``$\succeq$'' is referred to a comparison operator.  It is clear that $A$ is a PO allocation. Thus, the following will show that $A$ is exactly EF1.
	
	Suppose that $A$ is not EF1, there exists two agents $i, j\in N$,  $v(A_i)<v(A_j-g)$ for any $g\in A_j$. 
	If $|A_j| < |A_i|$, by the definition of $\mathcal{M}$, there exists $\{g_i^1, \ldots, g_i^l\}\subseteq A_i$ such that $A_j \cup \{g_i^1, \ldots, g_i^l\} \in \mathcal{F}$ and $|A_j\cup \{g_i^1, \ldots, g_i^l\}|=|A_i|$. Denote by $C= \{ g_0^1, \ldots, g_0^l\}$ an item set where every $g_0^k\in C$ is a copy item of $g_i^k\in A_i$ with $v_i(g_0^k) = v_j({g_0^k}) = 0$. Thus, based on Corollary~\ref{cor1}, there exists a bijection $\phi: A_i \to A_j \cup C$ such that $A_k - g + \phi(g) \in \mathcal{F}$ for any $k\in\{i,j\}$ and any $g\in A_k$. Due to the assumption of $v(A_i) < v(A_j)$, there must exist one pair $(g_i,\phi(g_i))$ such that $v(g_i) < v(\phi(g_i))$. Thus, we can construct another feasible allocation $B = (B_1,\ldots, B_n)$ where   $B_i=A_i-g_i+\phi(g_i)$, $B_j=A_j-\phi(g_i) + g_i$, $B_k=A_k$ for all $k\ne i,j$. Notice that $v(B_i) >v(A_i)$ and $v(B_j) \geq v(A_j - \phi(g_i)) > v(A_i)$. Thus, we have $A \prec B$ which is contradicted with $A\succeq B$. 
	
	For the remaining case of $|A_j| \geq |A_i|$, we can build similar contradictions by the above analysis. 
	Hence the leximin allocation $A$ is exactly EF1 and PO under matroid constraints. 
	\qed    
\end{proof}

The following theorem will focus on 2-valued valuations. Let $a > 1$ be a constant. For any agent $i \in N$, her valuation $v_i(\cdot)$ is 2-valued if $v_i(g) \in \{1, a\}$ for any $g\in M$.

\begin{theorem}\label{th3}
    For 2-valued valuations with $\{1, a\}$, every Max-NSW allocation guarantees $\max\{\frac{1}{a^2}, \frac{1}{2}\}$-EF1 and PO under matroid constraints.
\end{theorem}
\begin{proof}
	Given an instance $\mathcal{I} = (N, M, (v_i)_{i\in N})$ with a matroid $\mathcal{M} = (M, \mathcal{F})$. 
	$M=\{g_1, \ldots, g_m\}$ is a item set,
	$N=\{1, \ldots, n\}$ is a agent set  and $v_i(g)\in \{1,a\}$, $a>1$.
	Let $A = (A_1,\ldots, A_n)$ be a  Max-NSW allocation. 
	Due to Theorem~\ref{th1}, $A$ is $1/2$-EF1 and PO. The following will illustrate that $A$ is also $\frac{1}{p}$-EF1. 
	
	Suppose that $A$ is not $\frac{1}{a^2}$-EF1, there exists two agents $i, j\in N$,  $v(A_i)<\frac{1}{a^2}v(A_j-g)$ for any $g\in A_j$. 
	We will discuss from two cases depending on whether $|A_j| \le |A_i|$ or not.

	\noindent\textit{Case 1.} 
	If $|A_j| \le |A_i|$,  there exists one pair $(g_i,g_j)$ such that $A_i+g_j-g_i\in \mathcal{F}$, $A_j+g_i-g_j\in \mathcal{F}$, and $v_i(g_i) =1, v_i(g_j)=a$.  Construct another feasible allocation $B = (B_1,\ldots, B_n)$  as follows:
	$B_i=A_i-g_i+g_j$, $B_j=A_j-g_j+ g_i$, $B_k=A_k$ for all $k\ne i,j$. 
	For simplicity, denote $\Delta = NSW(B) - NSW(A)$. 
	The following illustrates that $NSW(B) > NSW(A)$, which contradicts the definition of $A$. 
	If $v_j(g_i)\le v_j(g_j)$, it holds that
	$\Delta
	=(a-1)v_j(A_j)-(v_j(g_j)-v_j(g_i))v_i(A_i) -(a-1)(v_j(g_j)-v_j(g_i)) 
	\geq (a-1)v_j(A_j)> 0$. 
	Otherwise, 
	$
	\Delta
	= (a-1)(v_j(A_j)-v_i(A_i)-(a-1))
	>(a-1)(v_j(A_j-g_j)-\frac{1}{a^2}v_i(A_j-g_j))
	>0
	$.

	\noindent\textit{Case 2.}      
	If $|A_j| > |A_i|$,  there exists $g_j\in A_j$ such that $A_i+g_j\in \mathcal{F}$.
	Construct another feasible allocation $B = (B_1,\ldots, B_n)$  as follows:
	$B_i=A_i+g_j$, $B_j=A_j-g_j$, $B_k=A_k$ for all $k\ne i,j$. 
	The following illustrates that $NSW(B) > NSW(A)$, which contradicts the definition of $A$. 
	In fact, if $v_i(g_j)=v_j(g_j)$,
	$
	\Delta
	=v_i(g_j)(v_j(A_j-g_j)-\frac{1}{a^2}v_i(A_j-g_j))>0.
	$
	If $v_i(g_j)>v_j(g_j)$,
	$
	\Delta
	=av_j(A_j)-v_i(A_i)-a
	>0
	$. 
	For the remaining case of $v_i(g_j)<v_j(g_j)$,
	$
	v_i(B_i)v_j(B_j)-v_i(A_i)v_j(A_j)
	=v_j(A_j)-av_i(A_i)-a
	>v_j(A_j-g_i)-\frac{1}{a}v_i(A_j-g_j)
	>0
	$.
	\qed
\end{proof}


\section{Beyond Matroid Constraints}
In this section, we proceed to study EF1 and PO allocations under constraints that go beyond matroids.
The classes of constraints we mainly focus on are $p$-extendible systems and independent systems.  

    

\subsection{$p$-Extendible System Constraints}

As mentioned previously,  a $p$-extendible system $(M, \mathcal{F})$ satisfies a key property that if $C\in \mathcal{F}, D \in \mathcal{F}$ with $C \subset D$ and if $x\notin C$ such that $C + x \in \mathcal{F}$, then there exists an item set $Y \subseteq D \setminus C$ with $|Y| \leq p$ such that $D \setminus Y + x \in \mathcal{F}$.
Notice that the swap-out set $Y$ is independent in $\mathcal{F}$ while adding a single item $x$. However, if adding several items $H$ to  $D$, $Y$ would not always admit independence. 

Building on the above observation, we will state a negative result for identical and binary valuations.


\begin{lemma}\label{lemma5}
    For identical binary valuations, there exists an instance, in which a Max-NSW allocation is exactly $\frac{2}{5}$-EF1 under 2-extendible system constraints.
\end{lemma}
\begin{proof}   
	Given a 2-extendible system $(M,\mathcal{F})$, where $M=\{g_1, \ldots, g_8\}$,
	$\mathcal{F}= P(\{g_1,g_2\}) \cup P(\{g_1,g_3,g_4,g_7, g_8\}) \cup P(\{g_2,g_3,g_4,g_7, g_8\}) \cup P(\{g_3,g_4,g_5,g_6, g_7, g_8\})$. Specifically, $P(S) =\{T: T \subseteq S\}$ is denoted by the power set of an item set $S$. 
	There are two agents with identical additive valuations $v(\cdot)$ and $v(g) =1$ for every single item $g\in M$.
	It is easy to verify that $\max_{X} \prod_{i=1}^{n}v_i(X_i)=12$, and $A_1=\{g_1,g_2\}$, $A_2=\{g_3, g_4, g_5, g_6, g_7, g_8\}$ is a Max-NSW allocation.
	Thus, $v(A_1) = 2 =\frac{2}{5} \cdot v(A_2 - g)$ for all $g\in A_2$.
	\qed
\end{proof}

In the 2-extendible system in Lemma~\ref{lemma5}, when we consider adding an independent set to another independent set, let $H=\{g_1,g_2\}$, $C=\emptyset$, $D=\{g_3, g_4, g_5, g_6, g_7, g_8\}$, there exists the unique subset $Y \subseteq D \setminus C$ with $Y=\{g_3, g_4, g_5, g_6, g_7, g_8\}$ such that $D \setminus Y \cup H\in \mathcal{F}$, but $|Y|>2\cdot 2$. 
Thus, we next introduce strongly $p$-extendible systems satisfying that the swapped-out $Y$ always satisfies $|Y|\le p|H|$.

\begin{definition}[Strongly $p$-extendible System]~\label{def7}
    A strongly $p$-extendible system is a pair $(E, \mathcal{F})$ where $\mathcal{F} \subseteq 2^{E}$ is a collection of subsets of $E$ which satisfies the following properties: (1) $\emptyset \in \mathcal{F}$; (2) if $D\in \mathcal{F}$ and $C\subseteq D$, then $C\in \mathcal{F}$; (3) if $C\in \mathcal{F}, D \in \mathcal{F}$ with $C \subset D$ and if $H\cap C=\emptyset$ such that $C \cup H\in \mathcal{F}$, then there exists a subset $Y \subseteq D \setminus C$ with $|Y| \leq p|H|$ such that $D \setminus Y \cup H\in \mathcal{F}$. 
\end{definition}

\begin{proposition}\label{prop1}
   The system $(E, \mathcal{F})$ is a matroid if and only if is strongly 1-extendible.
\end{proposition}

Now, we establish main results about strongly $p$-extendible systems. 

\begin{lemma}\label{lemma6}
For identical binary valuations, there exists an instance, in which, Max-NSW and EF1 are incompatible and all Max-NSW allocations are exactly $\frac{2}{3}$-EF1 under strongly 2-extendible system constraints.
\end{lemma}
  \begin{proof}   
	Given a 2-extendible system $(M,\mathcal{F})$, where $M=\{g_1, \ldots, g_8\}$,
	$\mathcal{F}= P(\{g_1,g_2\}) \cup P(\{g_1,g_3,g_4\}) \cup P(\{g_2,g_3,g_4\}) \cup P(\{g_3,g_4,g_5,g_6\})$. Specifically, $P(S) =\{T: T \subseteq S\}$ is denoted by the power set of an item set $S$. 
	There are two agents with identical additive valuations $v(\cdot)$ and $v(g) =1$ for every single item $g\in M$.
	It is easy to verify that $\max_{X} \prod_{i=1}^{n}v_i(X_i)=8$, and $A_1=\{g_1,g_2\}$, $A_2=\{g_3, g_4, g_5, g_6\}$ is the unique Max-NSW allocation.
	Thus, $v(A_1) = 2 =\frac{2}{3} \cdot v(A_2 - g)$ for all $g\in A_2$.
	\qed
\end{proof}
\begin{theorem}\label{th4}
For identical binary valuations, every Max-NSW allocation guarantees $\frac{1}{p}$-EF1 under strongly $p$-extendible system constraints. 

\end{theorem}
\begin{proof}
	Given an instance $\mathcal{I} = (N, M, v)$ with a $p$-extendible system $\mathcal{M} = (M, \mathcal{F})$. 
	$M=\{g_1, \ldots, g_m\}$ is a item set,
	$N=\{1, \ldots, n\}$ is a agent set and $v(\cdot)$ is a binary valuation for all agents. Let $A^*$ be a Max-NSW allocation.  
	Suppose, for contradiction, that $A^*$ is not $1/p$-EF1, and that there exist two agents $i,j\in N$ and any single item $g\in A^*_j$, 
	\begin{align*}\label{eq5}
		v(A^*_i) < \frac{1}{p} \cdot (v(A^*_j) - v(g)),~~~\forall g\in A^*_j.\tag{5}
	\end{align*}
	\noindent\textit{Case 1.} If there exists $g^* \in A^*_j$ such that $A^*_i + g^* \in \mathcal{F}$.
	Let $B$ be an alternative allocation such that $B_i = A_i^* + g^*$, $B_j = A_j^* - g^*$, and $B_k = A_k$ for all $k\neq i,j$. Due to $A_i^* + g^*\in \mathcal{F}$ and $A^*_k\in \mathcal{F}$ for all $k$, we get that $B$ is feasible under matroid $(M, \mathcal{F})$.
	Furthermore, $NSW(B) = v(B_i) v(B_j)\prod_{k\neq i,j} v(B_k) = NSW(A^*) + v(g)[v(A^*_j) - v(A^*_i) - v(g)] > NSW(A^*)$, where the inequality holds by Inequality (5). It contradicts the optimality of $NSW(A^*)$.

	\noindent\textit{Case 2.} If $A^*_i + g \notin \mathcal{F}$ for all $g \in A^*_j$. 
	By the definition of $p$-extendible system, for every item $ g \in A_i^*$, we can construct a feasible bundle $(A_j^* + g) \setminus Y_g\in \mathcal{F}$ by adding $g$ to $A_j^*$ and deleting some subset $Y_g \subset A_j^*$ with $|Y_g| \leq p$.
	Iteratively moving all items from $A^*_i$ to $A^*_j$, we can construct another feasible bundle by further deleting at most $p|A_i^*|$ items from $A_j^*$. Combining non-wastefulness of $A^*$, $|A_j^*| > p |A_i^*| + 1$. 
	Thus, there exists $Y\subset A_j^*$ with $|Y| \leq p |A_i^*|$ and $A_i^* \cup A_j^* \setminus Y \in \mathcal{F}$. 
	
	(i) If $|Y| > |A_i|$, then assume that $|Y| = |A^*_i| + d$.  Combining with $|Y| \leq p|A_i^*|$, $0 < d \leq (p-1)|A_i^*|$. 
	Let $B$ be an alternative allocation such that $B_i = Y$, $B_j = A^*_i + A_j^* - Y$, and $B_k = A_k^*$ for all $k$. Thus, $NSW(B) = NSW(A^*) + d (|A^*_j| - |A_i^*| - d) \geq NSW(A^*) + d (|A_j^*| - p|A_j^*|) > NSW(A^*)$, contradicting the optimality of $NSW(A^*)$.
	
	(ii) If $|Y| \leq |A_i|$, then we can construct an extension $C$ such that $Y \subset C \subset A_j^*$ and $|C| = |A_i^*| + 1$. Notice that $C \in \mathcal{F}$.
	Let $B$ be an alternative allocation where $B_i = C$, $B_j = A_i^* + A_j^* - C$, and $B_k = A_k^*$ for all $k$. 
	Thus, we obtain that $NSW(B) = NSW(A^*) + (|A_j^*| - |C|) > NSW(A^*)$ which contradicts the optimality of $NSW(A^*)$.
	\qed
\end{proof}



\subsection{Independent System Constraints}

This section focus on independence system constraints with two cases of valuations: general additive valuations and lexicographic preferences. We start with two negative results in additive settings.

\begin{lemma}\label{lemma7}
There exists an instance that no $(\frac{1}{4} + \varepsilon)$-EF1 and Max-NSW allocation under independence system constraints. 
\end{lemma}
\begin{proof}
	Given an independence system $(M, \mathcal{F})$, where $M = \{g_1, \ldots, g_m\}$ with $m = 1 + 1/\delta$, $\mathcal{F} = P(\{g_1\}) \cup P(\{g_2, \ldots, g_m\})$, and $\delta > 0$ be a small enough constant. There are two agents $\{1, 2\}$ that have identical valuation function $v(\cdot)$, and the values are as follows: $v(g_1) = 1 + \delta$, $v(g_k) = 4\delta$ for all single item $g_k \neq g_1$. 
	
	Consider an arbitrary feasible allocation $X$ under $(M,\mathcal{F})$. Due to identical valuation $v(\cdot)$, assume w.l.o.g. that $v(X_1) \leq v(X_2)$. 
	If $g_1 \in X_1 \cup X_2$, then $NSW(X) = v(X_1) v(X_2) \leq (1 + \delta) (4\delta \cdot 1/\delta) = 4(1 + \delta)$ since one agent has the only one large item $g_1$ and the other one have at most the remaining small items with value of $4\delta$. 
	Otherwise if $X_1 \cup X_2 \subseteq M - g_1$, then $NSW(X) \leq 2 \cdot (4\delta \cdot 1/(2\delta)) = 4$, and the equality holds when both agents have $1/(2\delta)$ small items. Thus, $NSW(X) \leq 4(1 + \delta)$ for all $X$. 
	Let $A^*$ be an allocation that $A^*_1 = \{g_1\}$ and $A^*_2 = \{g_2, \ldots, g_m\}$. By the definition of $(M,\mathcal{F})$, $A^*_1, A^*_2\in \mathcal{F}$ and thus $A^*$ is a feasible allocation. Therefore, $A^*$ is the unique Max-NSW allocation with  $NSW(A^*) = 4(1 + \delta)$. However, $v(A^*_1) = (1 - \delta^2) /4 \cdot v(A^*_2 - g) < 1/4\cdot v_1(A_2^* - g)$ for all $g\in A^*_2$. Thus, $A^*$ is not $(1/4 + \varepsilon)$-EF1 for any $\varepsilon > 0$. 
	\qed
\end{proof}

Next, we establish the main result about general additive valuations.

\begin{theorem}\label{th5}
For general additive valuations, every Max-NSW allocation guarantees $\frac{1}{4}$-EF1 and PO under independence system constraints.
\end{theorem}
\begin{proof}
	Given an independence system $(M, \mathcal{F})$ and any corresponding instance $\mathcal{I} = (N, M, (v_i)_{i\in N})$. 
	Let $A^*$ be a Max-NSW allocation returned by Algorithm 1. 
	We will analyze it in two situations as follows:

	\noindent\textit{Case 1.}
	$NSW(A^*) > 0$. In this case, $v_k(A_k^*) > 0$ for every agent $k$. 
	We assume that $A^*$ is not $1/4$-EF1, that is, there exist two agents $i,j\in N$ and any item $g\in A^*_j$ such that $v_j(A^*_j - g) > 4\cdot v_i(A_i^*)$. The following will construct an alternative allocation. 
	
	Consider agent $j$'s bundle $A_j^* = \{g_1,\ldots, g_t \}$, where $v_j(h_1) \leq \cdots \leq v_j(g_{t})$. Notice that the item $g_t$ has the maximum value in $A_j^*$. Define $H_1 = \{g_j \in A_j^* - g_t: j \equiv 1(mod~2) \}$ and similarly define $H_2 = \{g_j \in A_j^* - g_t: j \equiv 0 (mod~2)\}$. Due to $A_j^* \in \mathcal{F}$, $H_1, H_2 \in \mathcal{F}$. Furthermore, at least one of the following inequalities holds: i) $v_j(H_1) \leq v_j(H_2) \leq v_j(H_1 + g_t)$; ii)  $v_j(H_2) \leq v_j(H_1) \leq v_j(H_2 + g_t)$. Let $H_{max} = \arg\max \{v_i(H_1), v_i (H_2)\}$ be the preferable bundle for agent $i$.
	Then, we construct an alternative allocation $B$, where $B_i = H_{max}$, $B_j = A_j^* \setminus H_{max}$, and $B_k = A_k^*$ for all $k\neq i,j$. By the definition of independence system, $B$ is a feasible allocation under $(M,\mathcal{F})$. Therefore, $v_i(B_i) = v_i(H_{max}) \geq 1/2 \cdot v_i(A_j^* - h_t) > 2 \cdot v_i (A_i^*)$, $v_j(B_j) = v_j (A_j^*) -v_j(H_{max}) \geq 1/2 \cdot v_j(A^*_j)$, and thus $NSW(B) = v_i(B_i) v_j (B_j) \prod_{k\neq i,j} v_k(B_k) > v_i(A_i) v_j (A_j) \prod_{k\neq i,j} v_k(A_k) = NSW(A^*)$, which gives us the desire contradiction.

	\noindent\textit{Case 2.} $NSW(A^*) = 0$. 
	Define $N^+_{X} = \{k\in N, v_k(X_k) > 0 \}$ as the agents that obtain positive values from allocation $X$.
	Due to $v_i (\emptyset) = 0$, $v_i(A^*_j - g) > 0$, and thus $|A^*_j| \geq 2$. 
	Based on Algorithm 1, $v_j(A^*_j) > 0$ and $|N^+_{A^*}| \geq |N_{X}^+|$ for any Max-NSW allocation $X$.

	If agent $i$ also has positive valuation, i.e., $v_i(A^*_i) > 0$, we can construct an alternative allocation by similar method in Case 1, which towards a contradiction with the definition of $A^*$. 
	Otherwise if $i\notin N_{A^*}^+$, then $v_i (A^*_i) = 0$. Let $B$ be an alternative allocation where $B_i = A^*_j - g$, $B_j = \{g\}$, $B_k = A^*_k$ for all $k$, and item $g$ is an arbitrary item of $A^*_j$. By the definition of independence system, $B$ is a feasible allocation under $(M,\mathcal{F})$. Moreover, $v_i(B_i) = v_i(A^*_j - g) > 0$, $v_j(B_j) = v_j (\{g\}) > 0$, and $v_k (B_k) = v_k (A^*_k)$ for all $k$. Thus, $|N^+_{B}| > |N^+_{A^*}|$, contradicting the definition of $A^*$.
	\qed
\end{proof}
\begin{corollary}
For general additive valuations, every Max-NSW allocation guarantees $\max\{\frac{1}{p}, \frac{1}{4}\}$-EF1 and PO under strongly $p$-extendible system constraints. 
\end{corollary}

\subsection{Lexicographic Preferences}
Heretofore, we have studied additive valuations and  several special cases. Although these valuations crucial in practice, they are not worked on the scenes with ordering preferences.  
This section introduces lexicographic preferences.


\begin{algorithm}[h]
    \caption{Round-Robin under Independence System Constraints}\label{alg2}
    
    {\bf Input:} An IS $(M, \mathcal{F})$, $v_i(\cdot)_{i\in N}$, and an independence oracle. \\
    
    {\bf Output:} An allocation $A = (A_1,\ldots, A_{|N|})$. \\
    
    Initialize: $(A_1,\ldots, A_{|N|}) = (\emptyset,\ldots, \emptyset)$;\\
    
    \For{$r=1$ to $\lceil {\frac{|M|}{n}}\rceil$}
    {
    \For{$i=1$ to $|N|$}
    {
    \If{$\exists g\in M,$ s.t. $A_i +g \in \mathcal{F}$}
    {
        Let $g^*\in \arg \max_{g\in M} \{v_i(g): A_i + g\in\mathcal{F}\}$;\\
        $A_i=A_i+g^*$;\\
        $M=M-g^*$;
    }
\Else
{$N=N-i$;\\
}
}   
        
    }
    {\bf Return} $A = (A_1,\ldots, A_{|N|})$.   
\end{algorithm}
\begin{definition}[Lexicographic Preference]
Given an instance $(M, N, (v_i)_{i\in N})$. Every agent $i$ has a linear ordering $\pi_i = (g_{i1}, \ldots, g_{im})$ such that $v_i(g_{i1}) > \cdots > v_i(g_{im})$. Let $(o_S(g_{i1}),\ldots, o_S(g_{im}))$ be a vector of bundle $S$ allocated to agent $i$: $o_S(g_{ik}) = v(g_{ik})$ for all $g_{ik} \in S$ and $o_S(g_{ik}) = 0$ otherwise.
For any two bundles $X, Y \subseteq M$, agent $i$'s preference is induced by lexicographic order as follows:
\begin{equation*}
 v_i(X) > v_i (Y) \Leftrightarrow  \exists \ell, o_X(g_{i\ell}) > o_Y(g_{i\ell}) \mbox{~and~} \forall k <\ell,  o_X(g_{ik}) = o_Y(g_{ik}).
\end{equation*}
\end{definition}
To compute an exact EF1 and Max-NSW allocation, we present Algorithm~\ref{alg2} by Round-Robin method.  
In this algorithm, we make a slight modification of the classic Round-Robin that 
if any item is added to someone's bundle would break the independence requirement, he is being skipped over. 

Next, we provide a positive  result about lexicographic preferences. 
\begin{theorem}\label{th6}
For lexicographic preferences, the allocation computed by Algorithm~\ref{alg2} guarantees exactly EF1 and PO under independence system constrains. 
\end{theorem}
\begin{proof}
	Given an instance $\mathcal{I} = (N, M, (v_i)_{i\in N})$ with an independence system  $\mathcal{M} = (M, \mathcal{F})$, 
	where
	$M=\{g_1, \ldots, g_m\}$ is a item set,
	$N=\{1, \ldots, n\}$ is a agent set  and $v_i(g)\in \{1,a\}$, $a>1$.
	Let $A$ be the allocation returned by Algorithm~\ref{alg2}. 
	For each $t$-th iteration of outer for-loop, denote by $h_i^t$ the item allocated to agent $i$ in the $t$-th iteration. 
	
	To prove that $A$ is EF1, we consider arbitrary two agents $i, j \in N$. 
	By the procedure of Algorithm~\ref{alg2}, $v_i (h_i^t) \geq v_i (h_j^{t+1})$ for every iteration $1\leq t \leq \lceil {\frac{m}{n}}\rceil$. If $i < j$, then agent $i$ select items before agent $j$ in the first iteration. 
	Thus, $v_i(h_i^1) \geq v_i (g)$ for all $g\in A_j$. 
	By lexicographic preference, $v_i (A_i) \geq v_i (A_j)$.
	Otherwise if agent $j$ select her item before agent $i$ in each iteration. 
	Due to $v_i (h_i^t) \geq v_i (h_j^{t+1})$ for each iteration $t$, $v_i (A_i) \geq v_i (A_j - h_j^1)$. Hence, the output $A$ is exactly EF1.
	
	To prove that $A$ is PO, we suppose that $B$ is an alternative allocation such that $B_k\in \mathcal{F}$ for all $k$, where $v_d(B_d) >v_d (A_d)$ for some agent $d$, and $v_k (B_k) \geq v_k(A_k)$ for any other $k\neq d$. 
	Since $v_d(B_d) >v_d (A_d)$, 
	there must exist an item $o \in B_d \setminus A_d$ such that $v_d(o) > v_d(h_i^1)$. 
	By the procedure of Algorithm~\ref{alg2}, the item $o$ is allocated to an agent $c < d$, 
	which selects her item before $d$ in the first iteration. 
	Thus, $o = h_c^1 \in A_c \setminus B_c$. 
	Due to $v_c(B_c) \geq v_c(A_c)$, 
	there exists another item $o'\in B_c \setminus A_c$ satisfying  $v_c(o') > v_c(o)$. 
	Since that agent 1 obtain her most preferable item $h_1^1$, $v_1(h_1^1) \geq v_1(g)$ for all $g\in B_1$, we have $v_1(A_1) \geq v_1(B_1)$. 
	By repeatedly finding $o'$, there must exist an agent $1\leq a < d$ satisfying $v_a(A_a) > v_a(B_a)$. 
	It contradicts the assumption of $A$. Thus, the output $A$ is PO. 
	
	\qed
\end{proof}

Although Algorithm~\ref{alg2} performs perfect in lexicographic preferences, this algorithm performs unsatisfactory in additive valuations. 
\begin{lemma}\label{lemma8}
For additive valuations, there exists an instance that the allocation returned by Algorithm~2 may not guarantee $(\frac{1}{2} +\varepsilon)$-EF1 and PO under partition matroid constraints. 
\end{lemma}

\begin{proof}
Given an instance $\mathcal{I} = (N, M, (v_i)_{i\in N})$ with an independence system  $\mathcal{M} = (M, \mathcal{F})$. 
$M=\{g_1,\ldots,g_{3\eta}\}$ divided by $\eta$ categories $C_1=\{g_1,\ldots, g_{2\eta}\}$, $C_k = \{g_{2\eta + k}\}$ for any other $k \neq 1$,
$\mathcal{F} =\{S\subseteq M, |S \cap C_1| \leq \eta,|S \cap C_k| \leq 1 ~\text{for} ~k\neq 1\}$. 
There are two agents having identical additive valuation functions $v(\cdot)$: $v(g) = 1$ for all $g\in M$. 

Let $\eta > 1$ be a sufficiently large constant.
Execute the procedure of Algorithm~\ref{alg2}, agent 1 gets $\{g_{2\eta + 1}, \ldots g_{3\eta}\}$, agent 2 gets $\{g_1, \ldots, g_{\eta}\}$ in the first $\eta$ iterations. 
By the definition of $\mathcal{F}$, 
in each iteration $\eta < t \leq 2\eta$, agent 1 gets one more item $g_{\eta + t}$ but agent 2 cannot get any new item. 
Thus, $A_1 = \{ g_{\eta + 1}, \ldots, g_{3\eta}\}$ and $A_2 = \{g_1,\ldots, g_{\eta} \}$. It implies that $v(A_2) < (\frac{1}{2} + \frac{1}{2\eta - 1})\cdot v(A_1 - g)$. Hence, the lemma holds by setting $\varepsilon = \frac{1}{2\eta -1}$. 
\qed
\end{proof}

\section{Conclusions and Future Directions}
The tension between fairness and efficiency is a focal point of allocation problems; this tension appears even without computational considerations. 
When every allocation is feasible, Max-NSW allocations ensure EF1 and PO for additive valuations.
However, with constraints, there is an insurmountable gap between EF1 and Max-NSW, so the approximation analysis comes in handy. 
In this paper, we study the relatively broad constrained settings that agents have additive valuations and obtain
several tight results for approximate EF1 and Max-NSW allocations.
We show that, under matroid constraints, every Max-NSW allocation is $\frac{1}{2}$-EF1 and PO. 
Under beyond matroid constraints, the approximation ratio is relaxed to $\frac{1}{4}$ in other independence systems. 
When we set the constraints on strongly $p$-extendible systems, in which the swapped-out set $Y$ is independent, every Max-NSW allocation is $\max\{\frac{1}{p}, \frac{1}{4}\}$-EF1 and PO for identical binary valuations. As a matter of course, the most obvious open question is whether for additive valuations, $\frac{1}{p}$-EF1 and Max-NSW are compatible.

Other promising future directions are: (i) generalizing our results beyond additive valuations, and (ii) designing efficient algorithms to compute  approximate EF1 and PO allocations under matroid or independence system constraints. 
%
%

\bibliographystyle{splncs04}
\bibliography{sample-bibliography}

\appendix

\end{document}